\documentclass{article}[12]
\usepackage{amsmath}
\usepackage{amsfonts}
\usepackage{amssymb}
\usepackage{amsthm}
\usepackage{ifpdf}
\usepackage{subfig}
\usepackage{wrapfig}
\usepackage{fullpage}
\usepackage{cite}

\ifpdf

  \usepackage[pdftex]{epsfig}
  \usepackage[pdftex]{hyperref}

\else

    \usepackage[dvips]{epsfig}
    \newcommand{\href}[2]{#2}

\fi

\theoremstyle{definition}
\newtheorem{theorem}{Theorem}[section]

\newcommand{\Z}{\mathbb{Z}}

\newcommand{\pfunc}[3]{#1 : #2 \dashrightarrow #3 }

\newcommand{\color}{{\rm col}}

\newcommand{\strength}{{\rm str}}

\newcommand{\dom}{{\rm dom} \;}

\newcommand{\res}[1]{\textrm{res}\left(#1\right)}
\newcommand{\termasm}[1]{\mathcal{A}_{\Box}[\mathcal{#1}]}
\newcommand{\prodasm}[1]{\mathcal{A}[\mathcal{#1}]}

\begin{document}

\title{Efficient Squares and Turing Universality at Temperature 1 with a Unique Negative Glue}
\author{Matthew J. Patitz\footnote{ Department of Computer
Science, University of Texas--Pan American, Edinburg, TX, 78539, USA.
mpatitz@cs.panam.edu. This author's research was supported in part by National Science Foundation Grant CCF-1117672. }
\and Robert T. Schweller\footnote{ Department of Computer Science, University of Texas--Pan American, Edinburg, TX, 78539, USA. schwellerr@cs.panam.edu. This author's research was supported in part by National Science Foundation Grant CCF-1117672. }
\and Scott M. Summers\footnote{Department
of Computer Science and Software Engineering, University of Wisconsin--Platteville, Platteville, WI 53818, USA.
summerss@uwplatt.edu.}}
\date{}
\maketitle

\begin{abstract}
Is Winfree's abstract Tile Assembly Model (aTAM) ``powerful?'' Well, if certain tiles are required to ``cooperate'' in order to be able to bind to a growing tile assembly (a.k.a., temperature 2 self-assembly), then Turing universal computation and the efficient self-assembly of $N \times N$ squares is achievable in the aTAM (Rotemund and Winfree, STOC 2000). So yes, in a computational sense, the aTAM is quite powerful! However, if one completely removes this cooperativity condition (a.k.a., temperature 1 self-assembly), then the computational ``power'' of the aTAM (i.e., its ability to support Turing universal computation and the efficient self-assembly of $N \times N$ squares) becomes unknown. On the plus side, the aTAM, at temperature 1, is not only Turing universal but also supports the efficient self-assembly $N \times N$ squares if self-assembly is allowed to utilize three spatial dimensions (Fu, Schweller and Cook, SODA 2011). In this paper, we investigate the theoretical ``power'' of a seemingly simple, restrictive class of tile assembly systems (TASs) in which (1) the absolute value of every glue strength is 1, (2) there is a single negative strength glue type and (3) unequal glues cannot interact (i.e., glue functions must be ``diagonal''). We call this class of TASs the \emph{restricted glue} TASs (rgTAS). We achieve several positive results. First, we first show that the tile complexity of uniquely producing an $N \times N$ square with an rgTAS is $O\left(\frac{\log n}{\log \log n}\right)$, matching the upper and lower bounds for the aTAM in general. In another result, we prove that the rgTAS class is Turing universal in the aTAM by exhibiting a construction which simulates an arbitrary Turing machine.  Additionally, we provide results for a variation of the rgTAS class, partially restricted glue TASs, which is similar except that the magnitude of the negative glue's strength cannot be precisely controlled and can only assumed to be $\ge 1$.  These results consist of a construction with $O(\log n)$ tile complexity for building $N \times N$ squares, and a construction which simulates a Turing machine but with a greater scaling factor than for the rgTAS construction.

\end{abstract}

\section{Introduction}
\label{sec:introduction}

Even in an overly-simplified model such as Winfree's abstract Tile Assembly Model (aTAM) \cite{Winf98}, the theoretical power of algorithmic self-assembly is formidable.  Universal computation is achievable \cite{Winf98} and computable shapes self-assemble as efficiently as the limits of algorithmic information theory will allow  \cite{SolWin07, RotWin00, AdChGoHu01}.  However, these theoretical results all depend on an important system parameter, the temperature $\tau$, which specifies the minimum amount of binding force that a tile must experience in order to permanently bind to an assembly. The temperature $\tau$ is typically set to a value of $2$ because at this temperature (and above), the mechanism of ``cooperation'' is available, in which the correct positioning of multiple tiles is necessary before certain additional tiles can attach.  However, in temperature $1$ systems, where such cooperation is unenforceable, despite the fact that they have been extensively explored \cite{jLSAT1,CookFuSch11}, it remains an unproven conjecture that self-assembly at temperature $\tau < 2$ is incapable of universal computation.  It is also widely conjectured (most notably in \cite{RotWin00}), although similarly unproven, that the efficient self-assembly of such shapes even as simple as $N \times N$ squares is impossible.

Given the seeming theoretical weakness of tile assembly at temperature $1$, contrasted with its computational expressiveness at temperature $2$, it seems natural that experimentalists would focus their efforts on the latter.  However, as is often the case, what seems promising in theory is not necessarily as promising in practice.  It turns out that in laboratory implementations of tile assembly systems \cite{RoPaWi04,BarSchRotWin09,CheSchGoeWin07}, it has proven difficult to build true strength-$2$ glues in addition to being able to strictly enforce the temperature threshold (e.g. many errors that are due to ``insufficient attachment'' tend to occur in practice).  Therefore, the characterization of self-assembly at temperature 1 is quite worth pursuing.

With the goal in mind of specifying a model of self-assembly that is closer to the intersection of theoretical power and experimental plausibility, in this paper, we propose ``the aTAM at temperature $\tau=1 + \epsilon$''.  We introduce the class of \emph{restricted glue} tile assembly systems (rgTAS), which requires that (1) all glues have strength $-1$, $0$, or $1$, (2) that there is \textbf{only one} glue type that exhibits $-1$ strength (i.e., a repulsive force equivalent in magnitude to the binding force of a strength $1$ glue), and (3) the glue function is \emph{diagonal}, which means that a glue of one type interacts only with other glues of the same type.  Our goal in investigating rgTASs is to study the ``simplest'' systems of algorithmic self-assembly that retain the computational and geometrical expressiveness of temperature $2$ self-assembly.  In this paper, we achieve two positive results. First, we first show that the tile complexity of uniquely producing an $N \times N$ square in an rgTAS is $O\left(\frac{\log n}{\log \log n}\right)$.  This is especially notable since it matches the upper and lower bounds for the unrestricted temperature $2$ aTAM \cite{AdChGoHu01}. Furthermore, in a later result we prove that the class rgTAS is Turing universal.

The use of glues possessing negative strength values has been investigated within a variety of contexts \cite{jReifSahuPeng06,jDotKarMas10}.  However, previous results have been much less restrictive, allowing non-diagonal glue functions (meaning that glue types can have interactions, perhaps of different strengths, with multiple different glue types) and large magnitudes.  Additionally, no explicit bound has been set on the number of unique negative strength glue types.  In order to help bridge the gap between theory and experiment, we have proposed restrictions on tile assembly systems (in the form of the rgTAS class stated above).

Various experimental implementations of the Tile Assembly Model have utilized tiles created from DNA \cite{RoPaWi04,BarSchRotWin09,CheSchGoeWin07,WinLiuWenSee98,MaoSunSee99}. Moreover, several results have shown that magnetic particles can be attached to DNA molecules \cite{DNAandMagnets,DNAandMagnets2}.  Since two magnetized particles of the same polarity experience a repulsive force, the combination of DNA tiles with attached magnetic particles is a natural prospect for the implementation of negative strength glues.  It may be possible to adjust the size, composition, and position of the magnetic particles to cause the repulsive force experienced between two tiles to be roughly equal in magnitude to the attractive force experienced by two strength $1$ glues. (Note that utilizing magnetic polarity for glues has previously been modeled in \cite{MajumReif08}.) Also, the requirement of only a single negative glue type allows for the attachment of the same magnetized particle to any tile surface that needs to exhibit a $-1$ strength glue.

However, it remains possible that the means of providing the repulsive forces for negative glues cannot, in fact, be so finely tuned that they produce repulsive forces (nearly) exactly and consistently equal in magnitude to the positive strength-$1$ glues of tiles.  Therefore, we also present constructions for an adapted version of the rgTAS class, the \emph{partially restricted glue TAS} (or prgTAS) class, in which the negative strength glues can only be guaranteed to be of magnitude $\ge 1$, that is, they may repel with more force than $-1$.  In this case, the previous constructions no longer work correctly and we thus provide necessary adaptations.  It is notable that in this somewhat more relaxed class, our new constructions require additional tile complexity (for constructing $N \times N$ squares and simulating Turing machines) or larger scale factors (for simulating Turing machines) than the rgTAS constructions.

The organization of this paper is as follows.  In Section~\ref{sec:prelim}, we review the aTAM and define the rgTAS class, along with a few other definitions used in our subsequent constructions.  In Section~\ref{sec:squares}, we present the rgTAS and prgTAS constructions for self-assembling $N \times N$ squares.  In Section~\ref{sec:zig-zag}, we show how to simulate zig-zag systems (e.g., a TAS that simulates an arbitrary Turing machine on an arbitrary input) with an rgTAS and with a prgTAS.

\section{Preliminaries}\label{sec:prelim}

In this paper, we work in the $2$-dimensional discrete Euclidean space $\Z^2$.

Let $U_2 = \{(0,1), (1,0), (0,-1), (-1,0)\}$ be the set of all
\emph{unit vectors}, i.e., vectors of length 1 in $\mathbb{Z}^2$. We
write $[X]^2$ for the set of all $2$-element subsets of a set $X$.
All \emph{graphs} here are undirected graphs, i.e., ordered pairs $G
= (V, E)$, where $V$ is the set of \emph{vertices} and $E \subseteq
[V]^2$ is the set of \emph{edges}. All logarithms are base-$2$.

\subsection{The Abstract Tile Assembly Model}
\label{sec-tile-assembly-model}

We now give a brief and intuitive sketch of the Tile Assembly Model
that is adequate for reading this paper.  More formal details and
discussion may be found in \cite{Winf98,RotWin00,Roth01,jSSADST}.

Intuitively, a tile type $t$ is a unit square that can be
translated, but not rotated, having a well-defined ``side
$\vec{u}$'' for each $\vec{u} \in U_2$. Each side $\vec{u}$ of $t$
has a ``glue'' of ``color'' $\textmd{col}_t(\vec{u})$ -- a string
over some fixed alphabet $\Sigma$ -- and ``strength''
$\textmd{str}_t(\vec{u})$--an integer--specified by its type
$t$. Two tiles $t$ and $t'$ that are placed at the points $\vec{a}$
and $\vec{a}+\vec{u}$, respectively, interact with \emph{strength}
$\textmd{str}_t\left(\vec{u}\right)$ if and only if
$\left(\textmd{col}_t\left(\vec{u}\right),\textmd{str}_t\left(\vec{u}\right)\right)
=
\left(\textmd{col}_{t'}\left(-\vec{u}\right),\textmd{str}_{t'}\left(-\vec{u}\right)\right)$.
If $\textmd{str}_t\left(\vec{u}\right) > 0$, those tiles \emph{bind} with that strength.

Given a set $T$ of tile types, an \emph{assembly} is a partial
function $\pfunc{\alpha}{\Z^2}{T}$, with points $\vec{x}\in\Z^2$ at
which $\alpha(\vec{x})$ is undefined interpreted to be empty space,
so that $\dom \alpha$ is the set of points with tiles. $\alpha$ is
\emph{finite} if $|\dom \alpha|$ is finite. For assemblies $\alpha$
and $\alpha'$, we say that $\alpha$ is a \emph{subconfiguration} of
$\alpha'$, and write $\alpha \sqsubseteq \alpha'$, if $\dom \alpha
\subseteq \dom \alpha'$ and $\alpha(\vec{x}) = \alpha'(\vec{x})$ for
all $x \in \dom \alpha$.

A \emph{grid graph} is a graph $G =
(V,E)$ in which $V \subseteq \Z^2$ and every edge
$
\{
\vec{a},\vec{b}
\} \in E$ has the property that $\vec{a} - \vec{b} \in U_2$. The
\emph{binding graph of} an assembly $\alpha$ is the grid graph
$G_\alpha = (V, E )$, where $V =
\dom{\alpha}$, and $\{\vec{m}, \vec{n}\} \in E$ if and only if (1)
$\vec{m} - \vec{n} \in U_2$, (2)
$\color_{\alpha(\vec{m})}\left(\vec{n} - \vec{m}\right) =
\color_{\alpha(\vec{n})}\left(\vec{m} - \vec{n}\right)$, (3) $\strength_{\alpha(\vec{m})}\left(\vec{n} - \vec{m}\right) =
\strength_{\alpha(\vec{n})}\left(\vec{m} - \vec{n}\right)$, and (4)
$\strength_{\alpha(\vec{m})}\left(\vec{n} -\vec{m}\right) > 0$. An
assembly is $\tau$-\emph{stable}, where $\tau \in \mathbb{N}$, if it
cannot be broken up into smaller assemblies without breaking bonds
of total strength at least $\tau$ (i.e., if every cut of $G_\alpha$
cuts edges the sum of whose strengths is at least $\tau$).  For the
case of negative strength glues, we employ the model of irreversible
assembly as defined in \cite{jDotKarMas10}.

Self-assembly begins with a \emph{seed assembly} $\sigma$ (typically
assumed to be finite and $\tau$-stable) and
proceeds asynchronously and nondeterministically, with tiles
adsorbing one at a time to the existing assembly in any manner that
preserves stability at all times.

A \emph{tile assembly system} (\emph{TAS}) is an ordered triple
$\mathcal{T} = (T, \sigma, \tau)$, where $T$ is a finite set of tile
types, $\sigma$ is a seed assembly with finite domain, and $\tau$ is
the temperature. In subsequent sections of this paper, unless explicitly stated otherwise, we assume that $\tau = 1$ and $\sigma$ consists of a single seed tile type placed at the origin. An \emph{assembly sequence} in a TAS $\mathcal{T} = (T, \sigma, 1)$ is
a (possibly infinite) sequence $\vec{\alpha} = ( \alpha_i \mid 0
\leq i < k )$ of assemblies in which $\alpha_0 = \sigma$ and each
$\alpha_{i+1}$ is obtained from $\alpha_i$ by the ``$\tau$-stable''
addition of a single tile. The \emph{result} of an assembly sequence
$\vec{\alpha}$ is the unique assembly $\res{\vec{\alpha}}$
satisfying $\dom{\res{\vec{\alpha}}} = \bigcup_{0 \leq i <
k}{\dom{\alpha_i}}$ and, for each $0 \leq i < k$, $\alpha_i
\sqsubseteq \res{\vec{\alpha}}$.
If $\vec{\alpha} = \left(\alpha_i \mid 0 \leq i < k \right)$ is an assembly sequence in $\mathcal{T}$ and $\vec{m} \in \mathbb{Z}^2$, then the $\vec{\alpha}$-{\it index} of $\vec{m}$ is $i_{\vec{\alpha}}(\vec{m}) = \min\{ i\in \mathbb{N} \; \left| \; \vec{m} \in \dom{\alpha_i} \right. \}$. That is, the $\vec{\alpha}$-index of $\vec{m}$ is the time at which any tile is first placed at location $\vec{m}$ by $\vec{\alpha}$. For each location $\vec{m} \in \bigcup_{0\leq i < l}{\dom{\alpha_{i}}}$, define the set of its input sides $
\textmd{IN}^{\vec{\alpha}}(\vec{m}) = \left\{ \vec{u} \in U_2 \; \left| \;
\textmd{str}_{\alpha_{i_{\vec{\alpha}}(\vec{m})}}(\vec{u})>0
\right.\right\}$.

We write $\prodasm{T}$ for the
\emph{set of all producible assemblies of} $\mathcal{T}$. An
assembly $\alpha$ is \emph{terminal}, and we write $\alpha \in
\termasm{\mathcal{T}}$, if no tile can be stably added to it. We
write $\termasm{T}$ for the \emph{set of all terminal assemblies of
} $\mathcal{T}$. A TAS ${\mathcal T}$ is \emph{directed}, or
\emph{produces a unique assembly}, if it has exactly one terminal
assembly i.e., $|\termasm{T}| = 1$. A set $X$ \emph{strictly self-assembles} if there is a TAS $\mathcal{T}$ for which every assembly $\alpha\in\termasm{T}$ satisfies $\dom \alpha = X$.


\subsection{Restricted glue, zig-zag tile assembly systems and path simulation}
\textbf{Restricted Glue Tile Assembly Systems.} 

We say that a tile set $T$ is \emph{glue restricted} if (1) the absolute value of every glue strength in $T$ is 1, (2) the glue function is \emph{diagonal}, meaning that for every glue type $g$, the interaction between $g$ and any other glue type is of strength $0$, and of magnitude $1$ between two copies of $g$, and (3) there is a single negative-strength glue type. Intuitively, glue restricted tile sets are as ``close'' as one can get to \emph{pure} temperature one self-assembly in the aTAM. A TAS $\mathcal{T} = (T, \sigma, 1)$ is \emph{glue restricted} if $T$ is glue restricted. In this paper, for notational convenience, we will simply refer to a glue restricted TAS as an rgTAS.

Due to the potential difficulty of experimentally implementing negative strength glues such that the magnitude of their strengths is very nearly exactly equivalent to that of the positive strength-$1$ glues, we also introduce another type of tile assembly system, which we refer to as \emph{partially glue restricted} tile assembly system, or prgTAS.  A tile set in a prgTAS has the same properties as those for rgTASs, except for the fact that the magnitude of the strength of the single negative glue is guaranteed to be \emph{at least} 1, but may in fact be greater (i.e. the effective glue strength could actually be $-2, -3, $ etc.).

\textbf{Zig-zag tile assembly systems.} A tile assembly system $\mathcal{T} = (T,\sigma,2)$ is called a \emph{zig-zag} \cite{CookFuSch11} tile assembly system if (1) $\mathcal{T}$ is directed, (2) there is a single assembly sequence $\vec{\alpha}$ in $\mathcal{T}$, with $\termasm{T} = \{ \alpha \}$, and (3) for every $\vec{x} \in \dom{\alpha}$, $(0,1) \not \in \textmd{IN}^{\vec{\alpha}}(\vec{x})$.  Intuitively, zig-zag systems are those that grow one horizontal row at a time, alternating left-to-right and right-to-left growth, always adding new rows to the north and never growing south.  Zig-zag systems are capable of simulating universal Turing machines, and thus universal computation \cite{CookFuSch11}. If $\mathcal{T}$ is a zig-zag system with $\termasm{T} = \{ \alpha \}$ and for every $\vec{x} \in \dom{\alpha}$ and every $\vec{u} \in U_2$, $\textmd{str}_{\alpha(\vec{x})}(\vec{u}) + \textmd{str}_{\alpha(\vec{x})}(-\vec{u}) < 4$, then we say that $\mathcal{T}$ is \emph{compact}.  Intuitively, compact zig-zag systems are zig-zag systems that only extend the width of each row by one over the length of the previous row, and only grow upward by one tile before continuing horizontal growth. Compact zig-zag systems are capable of simulating universal Turing machines \cite{CookFuSch11}.


\textbf{Path simulation.} Let $\mathcal{T} = (T,\sigma,2)$ be a zig-zag TAS with assembly sequence $\vec{\alpha} = (\alpha_i \mid 0 \leq i < k)$.
We say that a restricted TAS $\mathcal{S} = (S,\sigma,1)$ \emph{path simulates} $\mathcal{T}$ (at \emph{scale factor} $c$) if (1) $\mathcal{S}$ has a single assembly sequence, (2) there exist computable indices $0 = i_{-1} < i_0 < \cdots < i_{k-1} < i_k = k$ satisfying $c = \max\left\{ \left. i_j - i_{j - 1} \; \right| \; 0 \leq j < k \right\}$ and (3) there exists a computable function $f: \mathcal{A}[\mathcal{S}] \longrightarrow \mathcal{A}[\mathcal{T}]$ such that, for all $0 \leq j < k$, $f\left( \res{\left( \alpha_{i_{j-1}}, \ldots, \alpha_{i_j - 1}\right)}\right) = \alpha_j$.
Intuitively, $\mathcal{S}$ path simulates $\mathcal{T}$ at scale factor $c$ if $\mathcal{S}$ uniquely produces a path of tiles that is logically divided into segments of length $c$, where each such segment corresponds to exactly one tile in $T$, and these segments self-assemble exactly in accordance with the unique assembly sequence of $\mathcal{T}$.
Note that the idea of one tile assembly system simulating another has been studied in other contexts as well \cite{USA, CookFuSch11}.

\section{Efficient Self-Assembly of Squares}\label{sec:squares}
The self-assembly of $N \times N$ squares has been studied extensively (see \cite{AdChGoHu01,Dot10,RotWin00,KS07,SFTSAFT,KaoSchS08}). Rothemund and Winfree conjectured in \cite{RotWin00} that, for every $N \in \mathbb{N}$, if $\mathcal{T}_N = (T_N, \sigma, 1)$ uniquely produces $S_N = \{0, 1, \ldots, N - 1\} \times \{0, 1, \ldots, N - 1\}$, then $|T_N| \geq 2N - 1$. In what follows, we show that this bound does not hold for glue restricted tile assembly systems or for partially restricted tile assembly systems.

\subsection{Optimal tile complexity with an rgTAS}\label{sec:opt-squares}
We first demonstrate an rgTAS construction for the self-assembly of any $N \times N$ square that achieves $O\left(\frac{\log N}{\log \log N}\right)$ tile complexity.  For almost all $N$, this meets the information theoretic lower bound of $\Omega\left(\frac{\log N}{\log \log N}\right)$ that holds for any unrestricted aTAM system.

\begin{theorem}
\label{thm:efficient_square} For all $N \in \mathbb{N}$, there exists an rgTAS $\mathcal{T} = (T_N, \sigma, 1)$, such that $S_N$ strictly self-assembles in $\mathcal{T}$, $\mathcal{T}$ is directed, and $|T_N| = O\left(\frac{\log N}{\log \log N}\right)$.
\end{theorem}

We now sketch our construction for Theorem~\ref{thm:efficient_square}.  Figure~\ref{fig:efficient_square_higher_level} shows a very high-level overview of the main components of the construction.  The general components are 1. a counter (whose base will be discussed later) which begins at a hard-coded value (also to be discussed) that counts in a zig-zag manner, with one row performing an increment and the next performing a check to see if the counter should halt, to its maximum value to form the left side of the square, 2. a set of zig-zag columns which pass a token diagonally downward and right from the top of the counter (stopping once that token reaches the bottom) to form the majority of the square, and 3. two sets of tiles which form short rows off of the bottom and right side of the square to make its dimensions exactly $N$.

\begin{figure}[htp]
    \begin{center}
    \includegraphics[width=3.5in]{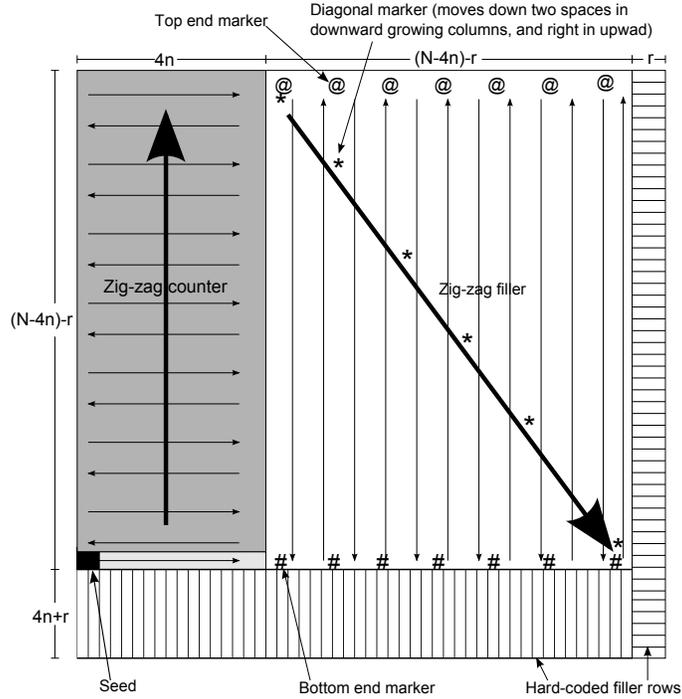} \caption{\label{fig:efficient_square_higher_level} \small A high-level overview of the construction of a tile type efficient $N \times N$ square.}
    \end{center}
\end{figure}

In order to mimic the cooperativity which occurs in temperature $2$ self-assembly (i.e. the attachment of a tile to an assembly via two strength-$1$ bonds), we simulate individual tiles using ``gadgets'' which are approximately $4 \times 4$ squares of tiles as pictured in Figure~\ref{fig:efficient_square_glue_example}.

\begin{figure}[htp]
    \begin{center}
    \includegraphics[width=4.5in]{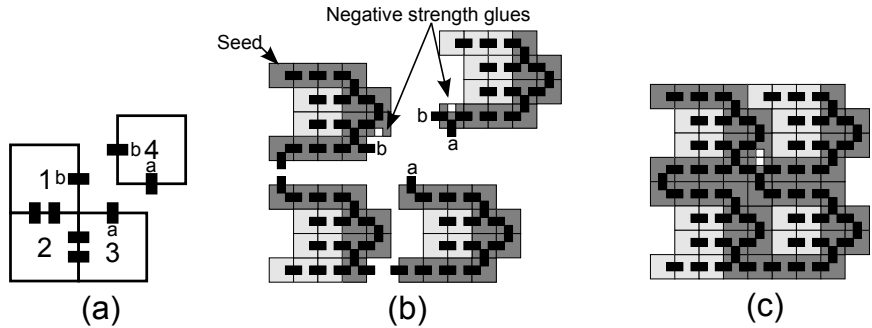} \caption{\label{fig:efficient_square_glue_example} \small An example using gadgets to simulate a temperature $2$ TAS (a) within an rgTAS (b).  Assume that the tiles in (a) attach in the order given by their numeric labels.  Each such tile is simulated in (b) by a roughly $4 \times 4$ square gadget of tiles.  Assembly in (b) begins with the top left tile of the gadget representing tile 1.  It then effectively proceeds by following the path through the dark grey tiles provided by neighboring bonds (the light grey tiles can fill at any point afterward).  The gadgets form one tile at a time along the dark grey path, and attached to each other, but are shown slightly separated in this figure to denote their boundaries.  Since an rgTAS is a temperature 1 system, strength-$1$ bonds in (b) can simulate the strength-$2$ bonds of (a).  However, to simulate the cooperativity of tile 4 attaching, a negative strength glue must be utilized.  Negative glues are represented as white squares on tile faces, and by contributing a $-1$ force to the strength of attachment, they ensure that the first tile in the gadget representing tile 4 must match both glues $a$ and $b$ in order for the sum of all three glue strengths to equal $1$ and allow the tile attachment which initiates the growth of the rest of the gadget.  It should also be noted that gadgets which use a pairing of a negative and a positive glue to form an ``output'' glue must enforce that the tile exposing the negative glue is placed first, to prevent potentially incorrect tile attachments to the positive glue.}
    \end{center}
\end{figure}

\begin{figure}[htp]
\begin{center}
    {\subfloat[{\small Gadgets which form the zig-zag counter.}]
    {\label{fig:efficient_square_counter_gadgets}\includegraphics[width=2.3in]{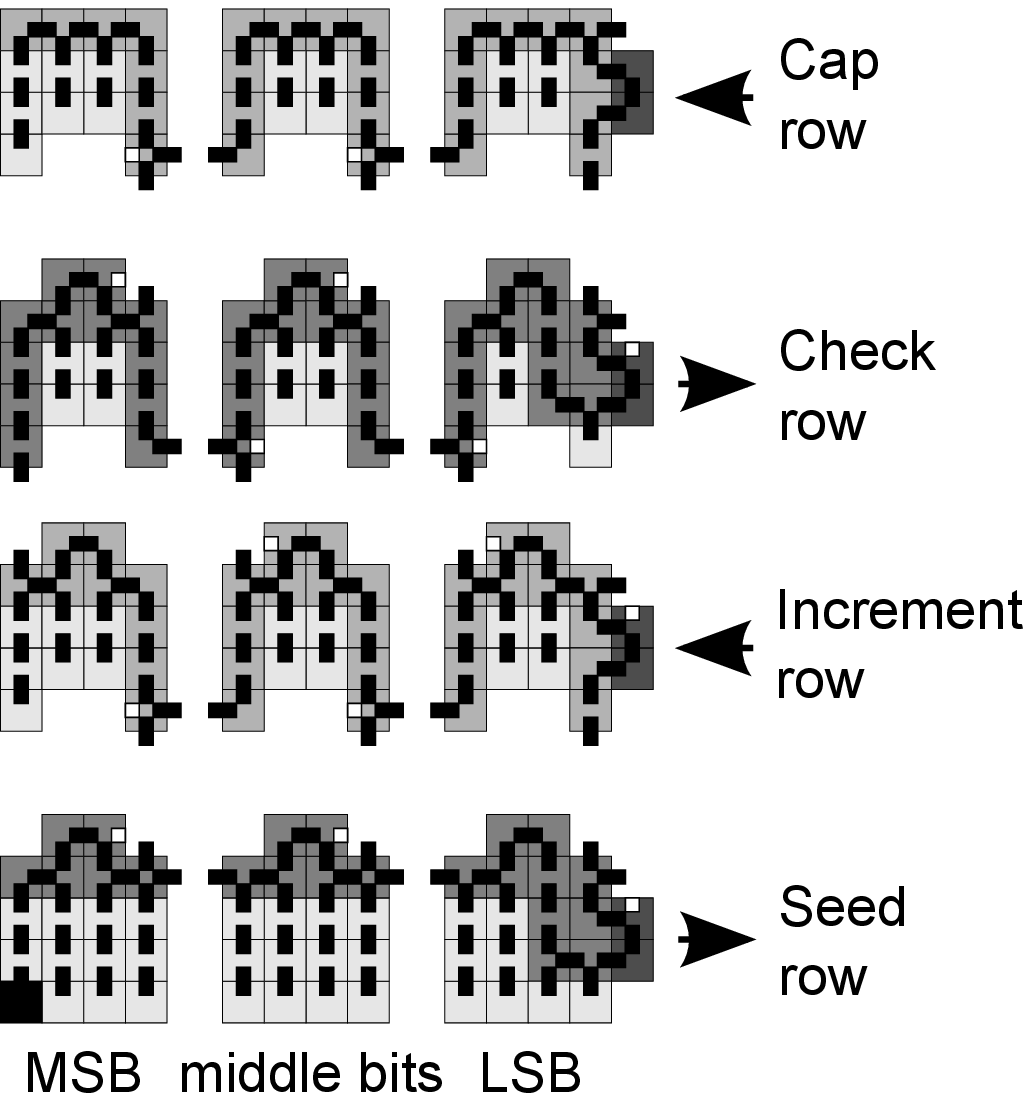}}}
    \quad
    {\subfloat[{\small Gadgets which form the zig-zag shifter.}]
    {\label{fig:efficient_square_shifter_gadgets}\includegraphics[width=2.0in]{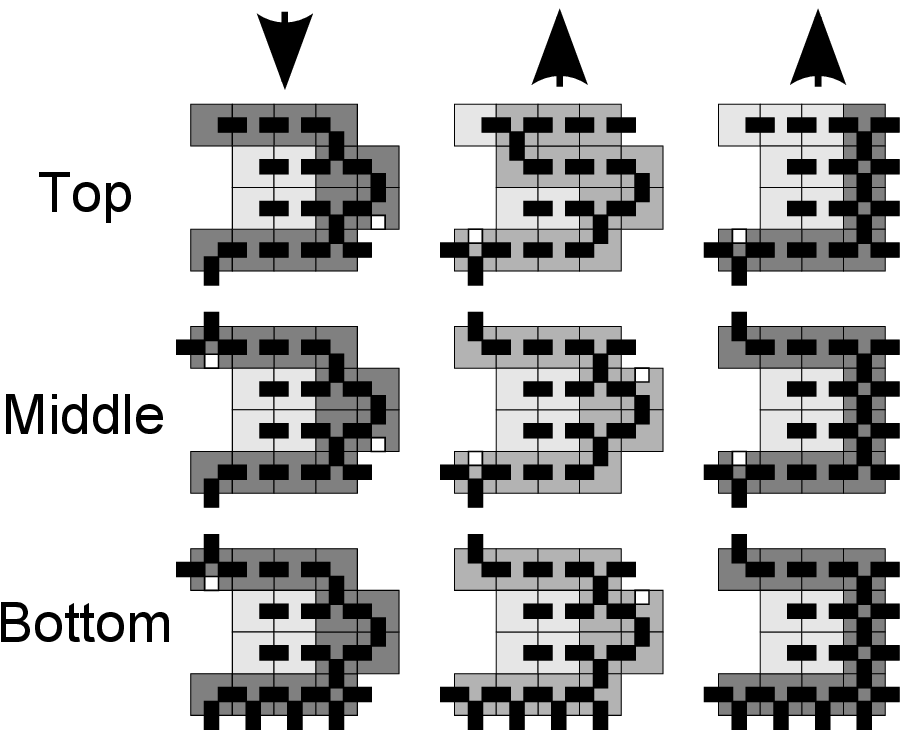}}}
    \vspace{-5pt}
    \caption{\small Main gadgets which form the tile type optimal $N \times N$ square.}
    \label{fig:efficient_square_gadgets}
    \vspace{-30pt}
\end{center}
\end{figure}

\begin{figure}[htp]
    \begin{center}
    \includegraphics[width=6.0in]{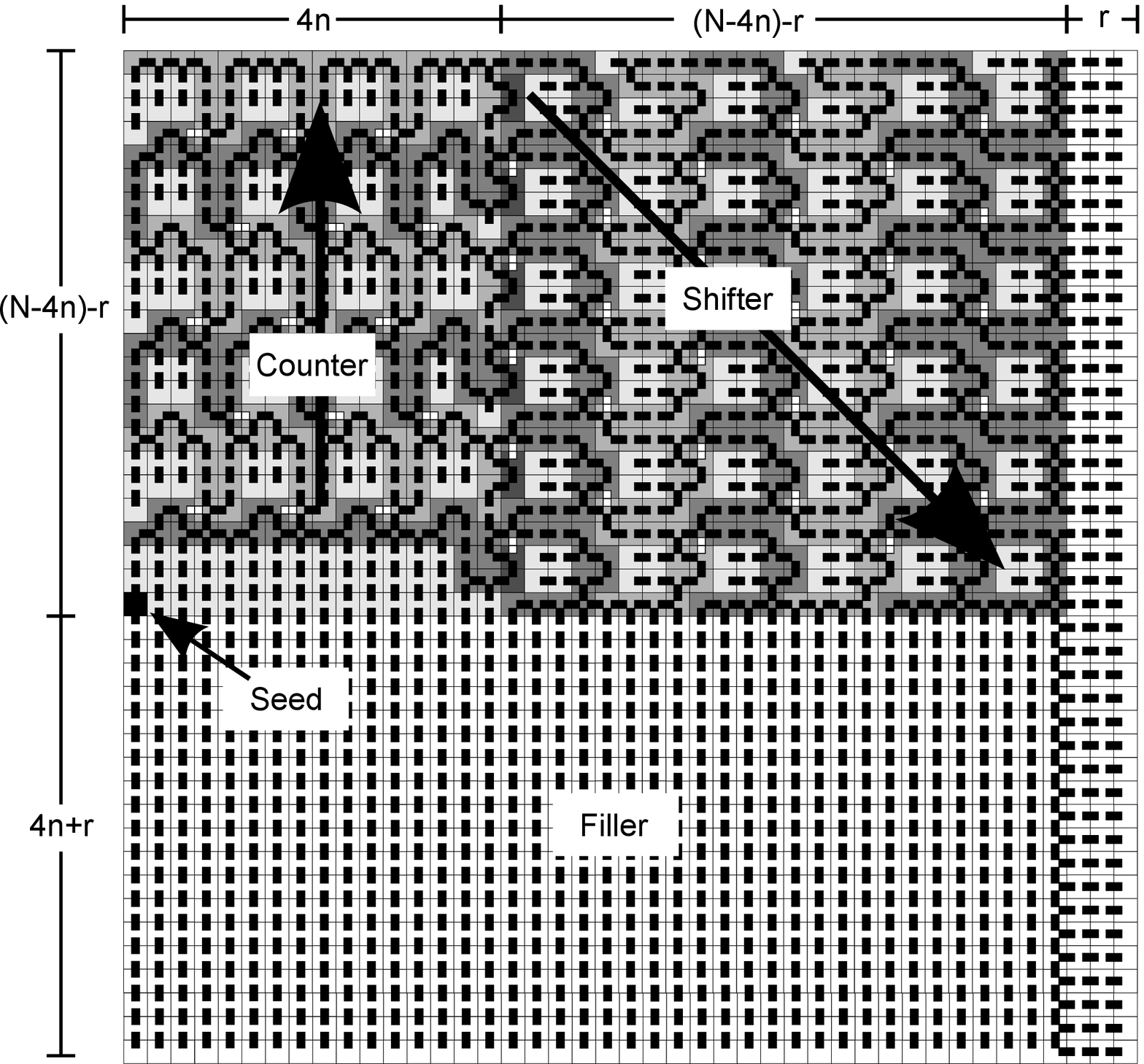} \caption{\label{fig:efficient_square_high_level} \small An construction of a tile type efficient $N \times N$ square, showing the structure of the gadgets. Note that the counter and shifter components are not shown to scale but would actually greatly dominate the square compared to the filler rows, with the shifter forming the majority.}
    \end{center}
\end{figure}

In order to achieve the optimal tile type complexity mentioned above, we make use of the technique introduced in \cite{AdChGoHu01} in which they convert a binary number $N$ into a base $b$ number, where $b$ is the power of $2$ such that $
   \frac{\log N}{\log \log N} \leq b = 2^k  < \frac{2\log N}{\log \log N}
$ for some positive integer $k$.  This ensures that the number of positions in the base-$b$ representation of $N$ is $n = \frac{\log N}{\log b} = O\left(\frac{\log N}{\log \log N}\right)$. Thus, since each gadget in our optimal square construction representing a position of that value requires a height of $4$ tiles, the number of rows of tiles necessary for an increment row--plus a check row is--$8$. Additionally, the width of the counter will be $4n$.  Let $c = \lfloor \frac{N-4n}{8} \rfloor$ be the number of pairs of rows of gadgets that the counter should assemble so that the diagonal shifter tiles can form the necessary dimensions.  Next, let $r = (N - 4n) \mod 8$ be the remainder necessary to pad the square out to exactly $N$.  Finally, let $s = b^{\lceil \log_b c \rceil} - c$ be the value at which the counter begins so that, by incrementing $c-1$ times, it will reach its maximum value of $b^{\lceil \log_b c \rceil}-1$. The value $s$ is logically encoded into the gadgets that grow from the seed to form the first row of the binary counter in our construction.

By allowing hard-coded rows of tiles of length $4n+r$ to attach to the bottom and rows of tiles of length $r$ to the right, our square construction exactly achieves the specified dimension of $N \times N$.

The tile complexity of the initial row of the counter is $O(n) = O\left(\frac{\log N}{\log \log N}\right)$.  The tile complexity of a base $b$ counter is $O(b) = O\left(\frac{\log N}{\log \log N}\right)$ and since the tile complexity blowup caused by creating the gadgets is $O(1)$, the tile complexity of the entire counter is $O\left(\frac{\log N}{\log \log N}\right)$.  Note that $O(1)$ tile types are required for the diagonal shifter component.  Finally, the bottom and right filler rows require $O(n) = O\left(\frac{\log N}{\log \log N}\right)$ and $O(1)$ tile types, respectively.  Therefore, the tile complexity of the overall construction is $O\left(\frac{\log N}{\log \log N}\right)$.


\subsection{$O(\log N)$ tile complexity with a prgTAS}\label{sec:nearly-optimal-squares}
We now present a construction using a prgTAS which is slightly less optimal than the previous rgTAS construction in terms of tile type complexity in that it self-assembles an $N \times N$ square using $O(\log N)$ tile types.

\begin{figure}[htp]
    \begin{center}
    \includegraphics[width=4.0in]{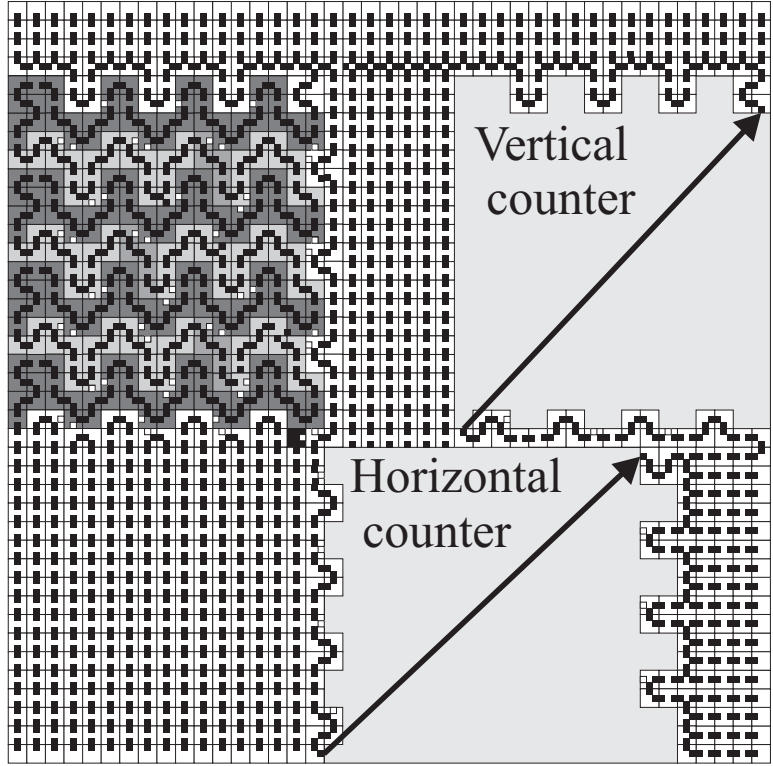} \caption{\label{fig:the_square} \small The negative glue is denoted as a little white square. Positive glues are denoted as little black squares. The plain white tiles represent filler tiles. Our construction is simple: it merely assembles a `U' shape via three counters and then fills in the ``interior'' of the `U' via generic filler tiles. In this example, $N = 41$, $n = 8$, $k = 3$, $K = 17$, $n_0 = 4$ and $x = 4$. The seed tile type is represented by the black square.}
    \end{center}
    \vspace{-20pt}
\end{figure}

\begin{theorem}
\label{square_theorem} For all $N \in \mathbb{N}$, there exists a restricted TAS $\mathcal{T} = (T_N, \sigma, 1)$, such that $S_N$ strictly self-assembles in $\mathcal{T}$, $\mathcal{T}$ is directed, and $|T_N| = O(\log N)$.
\end{theorem}

\begin{figure}[htp]
    \begin{center}
    \includegraphics[width=1.50in]{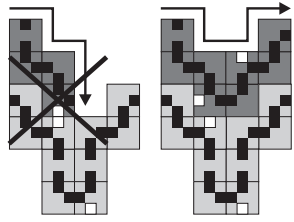} \caption{\label{fig:cooperation_gadget} \small We emulate ``temperature-2'' style cooperation by using geometry along with the careful placement of the negative glue in order to ensure that only the ``correct'' tile in a particular location of a path of tiles attaches and therefore can ``know'' if it is supposed to be, for example, a `1' or a `0' bit. We use this technique extensively throughout this paper. }
    \end{center}
\end{figure}

In what follows, we briefly sketch our construction for Theorem~\ref{square_theorem}. Let $n = \left\lfloor \frac{N - 1}{5} \right\rfloor$, $k = \left\lceil \log n \right\rceil$, $K = 5 + 4k$, $n_0 = 2^k - n + \left\lceil\frac{K}{5}\right\rceil$ and $x = N - \left(K + 5\left(n - \left\lceil\frac{K}{5}\right\rceil\right)\right)$. Intuitively, $N$ is the dimension (length of one side) of the target square $S_N$, $n$ is the number of count/increment row pairs that we will need in our construction, $k$ is the \emph{logical} width of the counter, $K$ is the \emph{actual} width of the counter in our construction, $n_0$ is the initial value for the counter, $2^k - 1$ is the maximum value of the counter and $x$ is the number of rows on top of the counter that we need to fill in with generic ``filler'' tiles. Although not necessarily surprising, it is worthy of note--and easy to show--that $x \leq 9$.
In Figure~\ref{fig:the_square}, we show a high-level overview of the terminal assembly produced by our construction (many details are omitted).

In our construction, we emulate the zig-zag counter of Rothemund and Winfree \cite{RotWin00}. We utilize three different binary counters in our construction, denoted as the \emph{first}, \emph{second} and \emph{third} counter and oriented vertically, horizontally and vertically respectively. We will discuss the general behavior of our north-growing zig-zag counter and highlight any subtle differences between the two other versions of it that we use in our construction.

The binary counter consists of a seed row, which encodes some number in binary, on top of which some number of increment/copy row pairs self-assemble in a zig-zag fashion. The top of the counter is capped off with a special cap gadget.

\begin{figure}[htp]
    \begin{center}
    \includegraphics[width=3.50in]{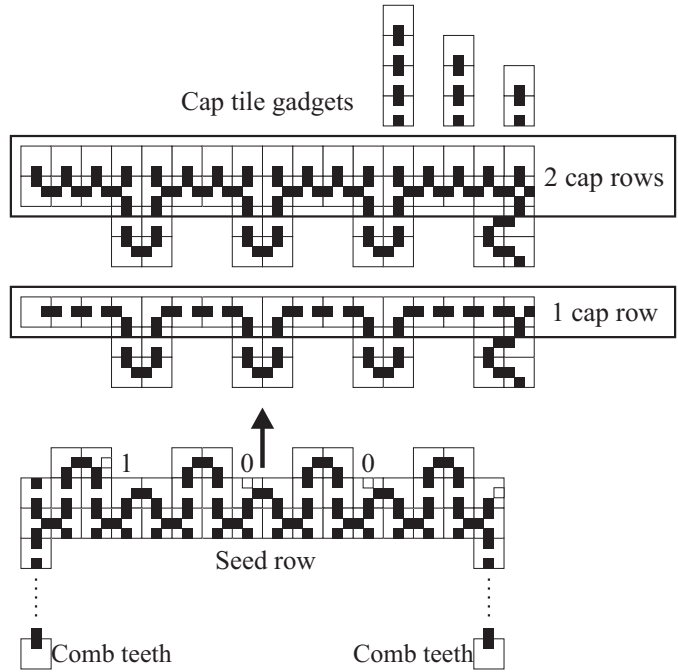} \caption{\label{fig:seed_row_cap_gadget} \small The bumps and dents along the north side of the seed row encode the initial value of the counter, denoted as $n_0$. In this example, $n_0 = 4$. The cap gadgets (of various sizes) are shown above the seed row. The number of cap tiles is $x$.}
    \end{center}
\end{figure}

\textbf{The Seed Row.} The seed row is a row of tiles that encodes the initial value of the binary counter $n_0$ using $k = \left\lceil \log n \right\rceil$ bits and has a horizontal extent of $K - 1$. The counter starts counting at this value and stops at $2^k - 1$. We encode the bits of $n_0$ via the careful placement of the negative glue (denoted as a little white square in all of our figures). The bit 0 is encoded by positioning the negative glue so that it is facing north in a dent and a 1 is encoded by positioning the negative glue so that it is facing east in a dent; see Figure~\ref{fig:seed_row_cap_gadget}. This bit encoding scheme is also used in count rows whereas slightly different encoding is used for copy rows. Off the bottom of the seed row, teeth of a ``comb'' attach in order to fill in the bottom left corner of the square. Each tooth has length $K$ and self-assembles to the south. The actual length of the seed row--and hence the actual width of our counter in this construction--is $K$.

\begin{figure}[htp]
    \begin{center}
    \includegraphics[width=3.0in]{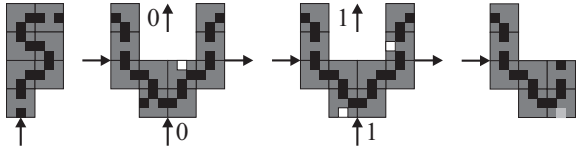} \caption{\label{fig:copy_row_bit_gadgets} \small Copy rows copy the bits advertised on the north side of the previous increment (or seed) row up for the next increment row. Copy rows encode each bit according to the ``mirror-image'' of the encoding utilized by the seed and increment rows. The lighter-grey square in the lower right tile represents a negative glue that \emph{may or may not} be present depending on whether or not the copy row is the first copy row to appear in the counter. These negative glues are used initiate the self-assembly of the second binary counter.}
    \end{center}
\end{figure}

\textbf{The Copy Rows.} Copy rows self-assemble on top of increment rows (including the seed row) from left to right and have horizontal extent $K$ (the actual width of the counter). Copy rows consist of a sequence of bit gadgets that utilize geometry and the careful placement of the unique negative (white) glue in order to emulate cooperations (see Figure~\ref{fig:cooperation_gadget}). In our construction, we have a one bit gadget for every bit in the binary representation of $n_0$ (this information is encoded directly into the bit gadget so that it knows which bit it is, e.g., most significant, least significant, third, etc). The bit gadgets that comprise each copy row are shown in Figure~\ref{fig:copy_row_bit_gadgets}. In our construction, if a copy row reads a string of 1 bits, i.e., $2^m - 1$ for some $m \in \mathbb{N}$, it will terminate the counter and allow the cap gadget to attach.

\textbf{The Increment Rows.} Each increment row increments the value of the counter by 1. Increment rows self-assemble from right to left (compared to left to right for copy rows--hence the zig-zag nature of our counter). Similar to copy rows, increment rows consist of a sequence of (a different type of) bit gadgets that each know ``which'' bit they represent. The bit gadgets for increment rows are shown in Figure~\ref{fig:increment_row_bit_gadgets}.

\begin{figure}[htp]
    \begin{center}
    \includegraphics[width=1.75in]{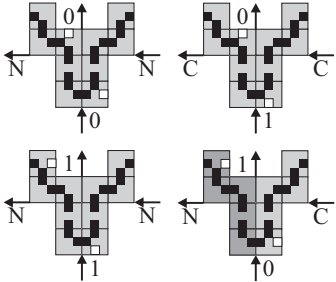} \caption{\label{fig:increment_row_bit_gadgets} \small The increment row bit gadgets read the bits of the previous copy row. The bit gadgets for increment rows emulate the standard binary counter tile types, such as those depicted in Figure 1 of \cite{RotWin00}. For each bit in the binary representation of $n_0$, we have four types of bit gadgets. The inputs are always south (0 or 1 bit value) and east (carry/no-carry).}
    \end{center}
\end{figure}

The value of the final increment row in our counter is $2^k - 1$ giving the counter an actual height of $5\left(n - \left\lceil \frac{K}{5} \right\rceil\right)$ rows of tiles. The bit pattern of $2^k - 1$ is detected by the (final) copy row, which terminates the counting. On top of the final copy row of the counter, a special cap tile gadget attaches, which is a path of tiles that fills in--and smooths out--the top of the jagged zig-zag counter. For each value of $x$ (ranging from 1 to 9), we use a different cap tile gadget. The top portion of Figure~\ref{fig:seed_row_cap_gadget} shows the two types of cap gadgets that we use in our construction--one allows additional ``comb teeth'' (each of varying height/length depending on $x$) to attach and the other that simply caps the counter.

\textbf{Completing the Square.} After--and only after--the first binary counter completes, may the construction proceed. To the upper right corner of the first binary counter, a path of tiles crawls down along the right side of the first vertical counter toward the seed tile. This path of tiles detects the seed tile via the south-facing negative glue in the lower rightmost tile in each copy row. A south-facing negative glue tells this path of tiles to ``keep going.'' Only the black seed tile type has an east-facing negative glue, which tells the path to ``stop'' and build the seed row for the second (horizontal) counter. Note that we do not encode any location information into these tiles that crawl down the right side of the first counter, which means that there are $O(1)$ such tile types participating in the formation of the path.

The second binary counter (the base of the `U' backbone) behaves similarly to the first counter except its top (logically, its least significant bit) is completely smooth so as to allow the seed row of the third and final (vertical) binary counter to attach. Furthermore, the cap tile gadget for the second counter places one more row of cap tiles on top of (actually, to the right of) the second counter than the cap tile gadget did for the first counter to ensure that the terminal structure is a square. The cap tile gadget for the second counter also initiates the self-assembly of the seed row for the third--and final--binary counter.

\begin{wrapfigure}{l}{3.00in}
    \begin{center}
    \includegraphics[width=2.50in]{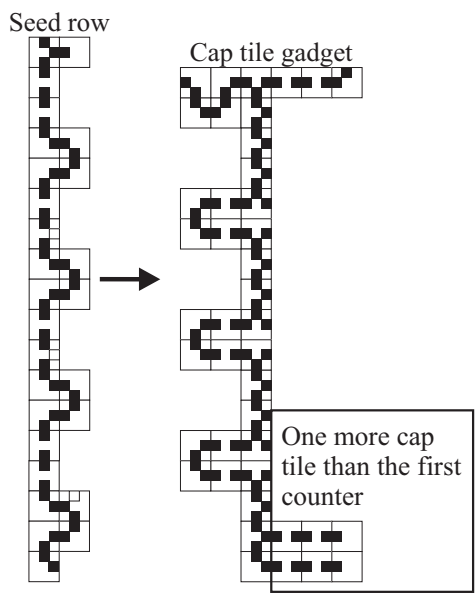} \caption{\label{fig:second_seed_cap_gadget} \small The cap tile gadget ensures that there is one additional row of cap tiles to compensate for the relative positions of the first two counters. }
    \end{center}
\end{wrapfigure}

The third (and final) counter, which happens to grow in a vertical fashion, in our construction completes the `U'-shaped backbone of the nearly-optimal square construction. This counter behaves similar to the first counter except its right edge is completely smooth. We also use a third type of cap tile gadget to form the smooth top of the square. This third type of cap gadget allows comb teeth (whose size depends on $x$) to bind to its north side and also shoots a path of tiles off to the left and back toward the first counter.

This path of tiles is eventually blocked by the first counter, but as this path self-assembles to the left, it allows filler tiles to fill in the interior of the square (see Figure~\ref{fig:third_seed_cap_gadget} for an example) and comb teeth to attach on top.

\textbf{Tile Complexity.} We use $O(1)$ generic filler tiles (white tile types in our figures) that either attach on top of (or to the right of) cap rows or fill in the interior of the square. There are $O(1)$ tiles that crawl down the right side of the first binary counter. The filler tiles that fill in the bottom left corner of the square must grow to length $O(K) = O(k)$ and stop for which $O(k)$ unique tile types suffice. Finally, we must encode the appropriate bit location into every bit gadget of the seed row and every copy, increment and cap tile gadget. Since there are $O(1)$ types of bit gadgets for each row and $k$ bit locations in each of the three different counters that we use, the tile complexity of our construction is dominated by $O(k) = O(\log N)$.

\begin{figure}[htp]
    \begin{center}
    \includegraphics[width=4.5in]{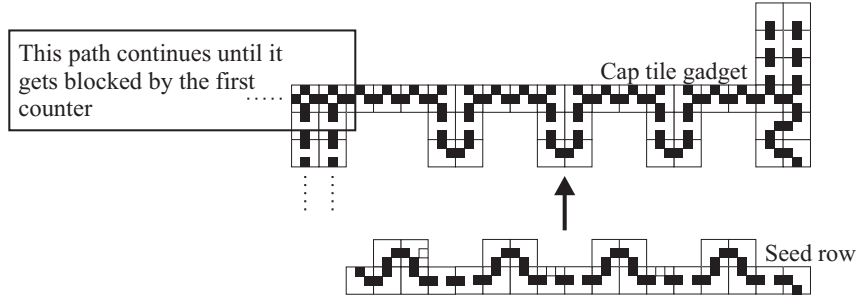} \caption{\label{fig:third_seed_cap_gadget} \small The cap gadget for the third counter shoots a path of tiles in the direction of--and is blocked by--the first counter. This ``off-shoot'' allows any necessary cap tiles to attach on its top and it also initiates the self-assembly of the interior of the square via generic filler tiles.}
    \end{center}
\end{figure}

It is interesting to note that our nearly-optimal construction (NOC) is essentially a spanning tree, much like the ``$2N - 1$'' construction of \cite{RotWin00}. However, in our NOC, we are allowed to use a single negative glue type, which--in conjunction with some clever use of geometry--allows us to emulate the cooperativity of ``temperature 2'' self-assembly. Furthermore, the longest simple path of tiles in Rothemund and Winfree's ``$2N - 1$'' construction is $2N - 1 = O(N)$ whereas the longest simple path of tiles in our NOC is $O\left(N \log N\right)$ but our NOC ensures that the length of every simple (un-blocked) path of tiles cannot exceed $O(\log N)$ without encountering the negative glue. 
\section{Turing Universality}\label{sec:zig-zag}

In this section, we show that for every zig-zag TAS in the aTAM, there is an rgTAS and a prgTAS that simulates it.  The simulation by an rgTAS requires only a constant factor increase in tile complexity and a constant size increase in scale factor over the simulated system.  The simulation by a prgTAS, on the other hand, requires an asymptotically similar increase in tile complexity, but an increase in scale factor that grows as the $\log$ of the size of the tile set of the TAS being simulated.

\subsection{Compact zig-zag simulation with an rgTAS}\label{sec:zig-zag-rgtas}
\begin{theorem}
\label{thm:zigzagSimulation1NegGluergTAS}
For every compact zig-zag TAS $\mathcal{T} = (T,\sigma,2)$, there exists an rgTAS $\mathcal{S} = (S,\gamma,1)$ such that $\mathcal{S}$ path simulates $\mathcal{T}$ at scale factor $12$ with $|S| = O( |T|)$.
\end{theorem}

The remainder of this subsection is devoted to a brief, intuitive sketch of our construction for Theorem~\ref{thm:zigzagSimulation1NegGluergTAS}.

Intuitively, $\mathcal{S}$ simulates $\mathcal{T}$ by logically converting each tile type $t \in T$ into a group of tile types in $S$ that self-assemble into a \emph{macro-tile}.  Since each macro-tile is the same size ($12$ tiles), it is easy to compute the indices $i_{-1} < i_0 < \cdots i_{k}$ in the definition of path simulate.

As shown in Figure~\ref{fig:compact-zig-zag-gadgets}, each $t \in T$ which can be placed in a growing compact zig-zag assembly has well-defined ``input'' and ``output'' sides (for convenience, and without loss of generality, we fix the seed row as growing from left to right). The corresponding macro-tiles are also depicted in Figure~\ref{fig:compact-zig-zag-gadgets}.  (Note that the particular shapes and sizes of the macro-tiles are designed so that all macro-tiles have the same number of tiles and also exactly one path of assembly, in order to correspond to the definition of path simulate.)  Double strength bonds on $t$ are simply simulated with a single strength-$1$ glue in the proper position of its macro-tile since $\mathcal{S}$ is a temperature $1$ system. Strength-$1$ glues are simulated by ensuring that every time a glue simulating a strength-$1$ glue is placed, that the adjacent location - in which a tile would be placed to bind to that glue - already also has adjacent to it a singly copy of the negative strength glue.  This ordering is guaranteed by the careful design of the macro-tiles and their order of growth.  In this way, whenever a tile would be placed in $\mathcal{T}$ by attaching to exactly two strength-$1$ glues, in $\mathcal{S}$ a tile is placed which binds to two strength-$1$ glues and is also repelled by a single negative strength-$1$ glue, for a total binding force of strength $1$.  This provides a mechanism for cooperation by ensuring that both input glues are matched.

Figure~\ref{fig:compact-zig-zag-simulation} gives an example of how the growth of the first few rows of $\mathcal{T}$ are simulated by the growth of connected gadgets in $\mathcal{S}$.  It is clear that $\mathcal{S}$ will exactly path simulate $\mathcal{T}$ at scale factor $12$.  Since $\mathcal{T}$ is an arbitrary compact zig-zag TAS and $\mathcal{S}$ is an rgTAS, Theorem~\ref{thm:zigzagSimulation1NegGluergTAS} is proven.

\begin{figure}[htp]
    \begin{center}
    \includegraphics[width=6.0in]{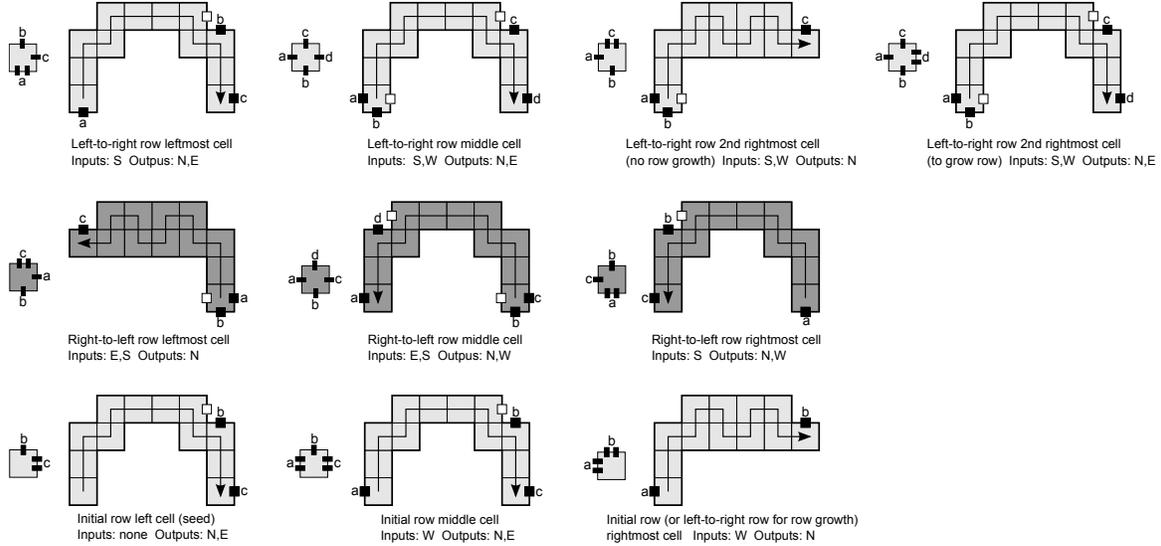} \caption{\label{fig:compact-zig-zag-gadgets} \small The individual tiles which attach to a growing compact zig-zag assembly (pictured as individual tiles on the left) and their corresponding macro-tiles (pictured to their right) which form to simulate those logically different tiles, based on their input and output sides.}
    \end{center}
    \vspace{-20pt}
\end{figure}

\begin{figure}[htp]
    \begin{center}
    \includegraphics[width=5.5in]{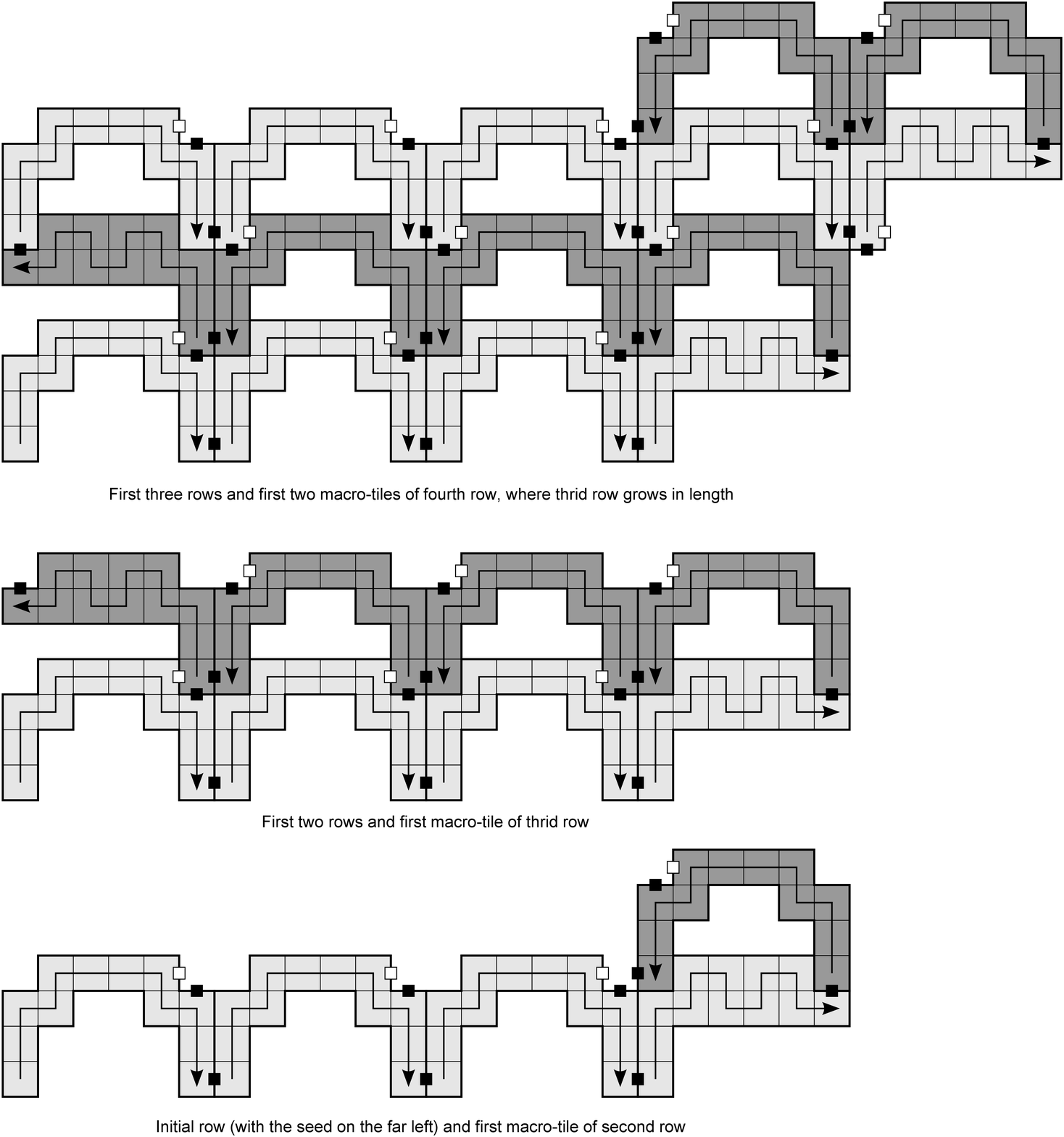} \caption{\label{fig:compact-zig-zag-simulation} \small An example of macro-tiles assembling to simulate the first few rows of a temperature $2$ compact zig-zag TAS using an rgTAS.}
    \end{center}
    \vspace{-20pt}
\end{figure}

\subsection{(Less) compact zig-zag simulation with a prgTAS}\label{sec:zig-zag-rgtas}
\begin{theorem}
\label{thm:zigzagSimulation1NegGlueprgTAS}
For every compact zig-zag TAS $\mathcal{T} = (T,\sigma,2)$, there exists a prgTAS $\mathcal{S} = (S,\gamma,1)$ such that $\mathcal{S}$ path simulates $\mathcal{T}$ at scale factor $O(\log |T|)$ with $|S| = O( |T|)$.
\end{theorem}

\begin{figure}[htp]
\begin{center}
    {\subfloat[{\small The set of input/output side combinations (grouped by input sides) for a zig-zag TAS.  Note that the right side of a seed row utilizes the the rightmost left-to-right type.}]
    {\label{fig:zig-zag-TM-input-output}\includegraphics[width=1.3in]{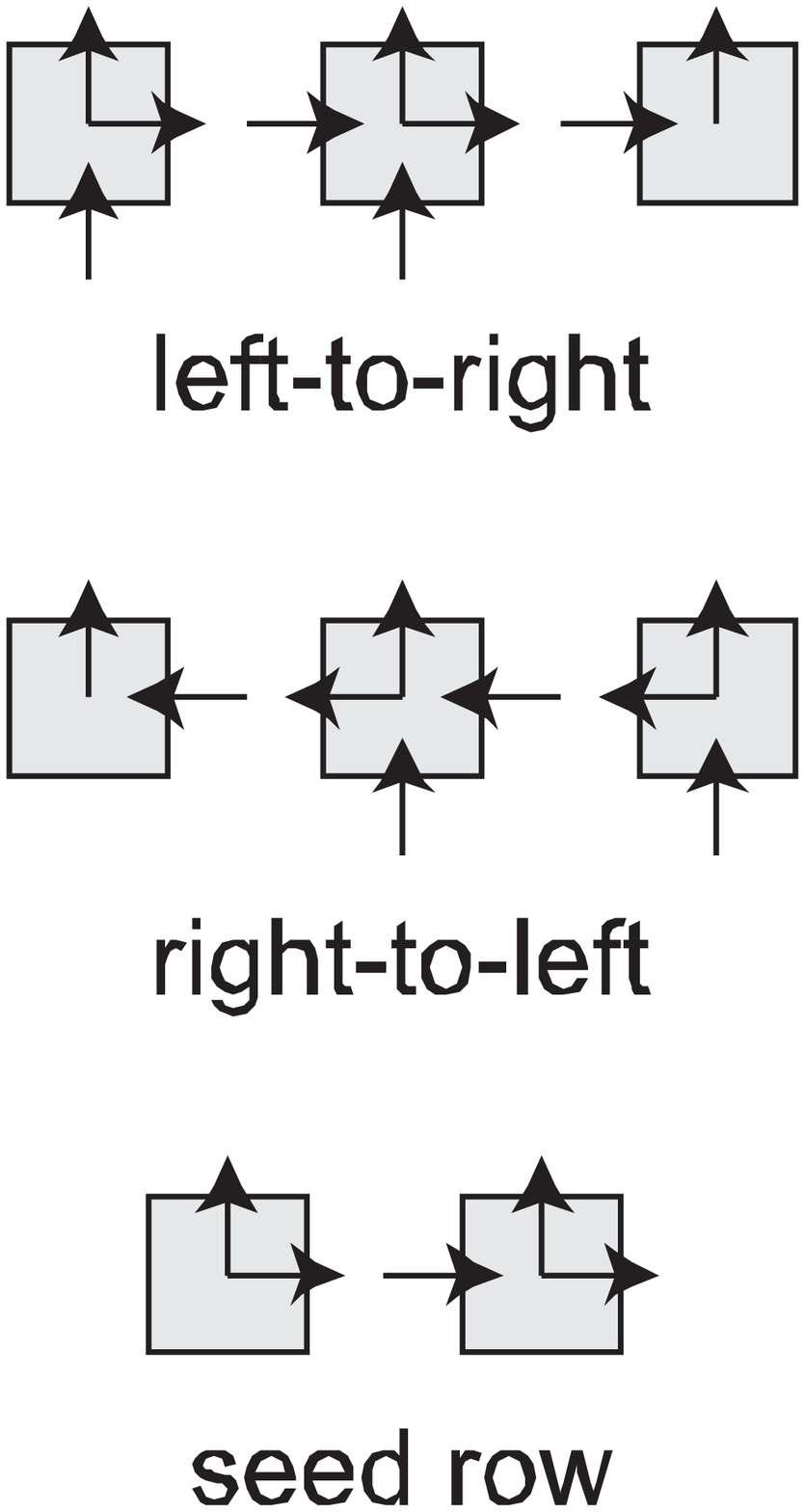}}}
    \quad
    {\subfloat[{\small Macro-tiles that simulate individual tiles from the zig-zag TAS.  The arrows show the direction of growth and the schematic tiles on the right show which directions are input and output sides for each macro-tile (with those in parenthesis represented by mirror images of the macro-tiles). The checkered tiles represent locations at which a negative glue is placed in order to tell the light(est) grey path to ``keep going'' to the right--similar to the tiles in the construction of Theorem~\ref{square_theorem} that crawl down the right side of the first binary counter. }]
    {\label{fig:zig-zag-macro-tiles}\includegraphics[width=3.0in]{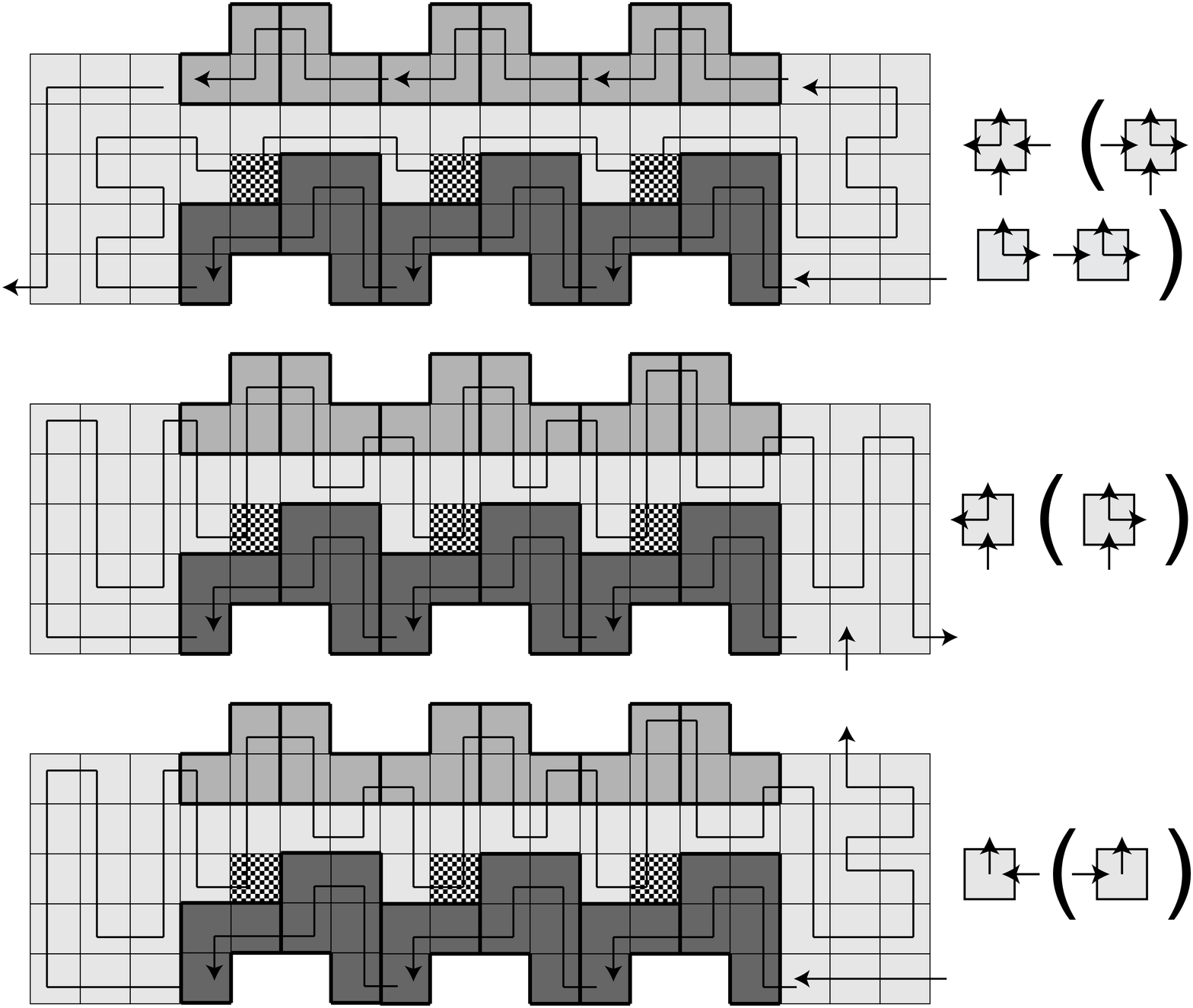}}}
    \quad
    {\subfloat[{\small Reading one bit of the glue type.}]
    {\label{fig:zig-zag-read-bit}\includegraphics[width=1.1in]{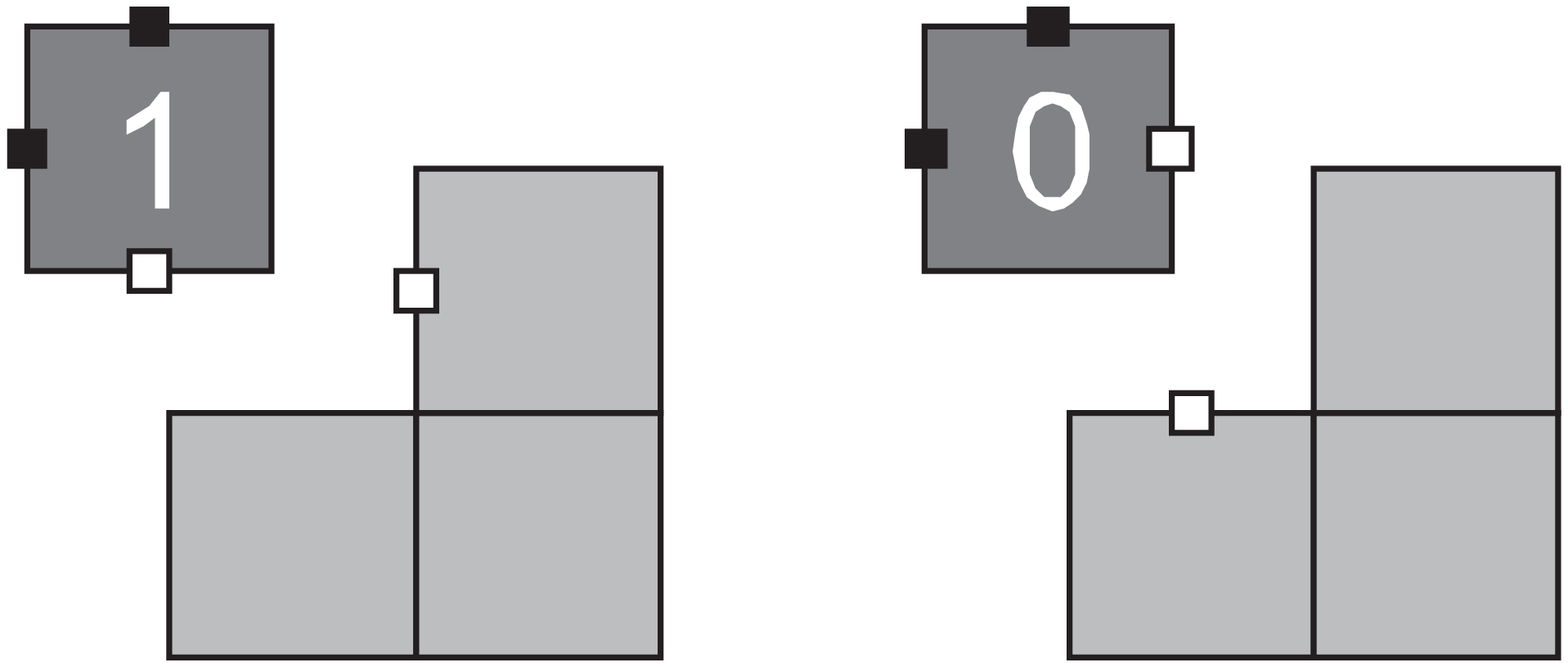}}}
    \vspace{-5pt}
    \caption{\small Details of the zig-zag simulation construction.}
    \label{fig:zig-zag-simulation}
\end{center}
\end{figure}

The remainder of this subsection is devoted to a brief, intuitive sketch of our construction for Theorem~\ref{thm:zigzagSimulation1NegGlueprgTAS}.

Intuitively, $\mathcal{S}$ simulates $\mathcal{T}$ by logically converting each tile type $t \in T$ into a group (path) of tile types in $S$ that self-assemble into a \emph{macro-tile}.  Let $G$ be the number of unique strength-$1$ north/south glues in $T$. Each macro-tile is a path of $2 \left\lceil \log G\right\rceil + 30$ tiles and forms in its entirety before allowing the next macro-tile to form, whence the scale factor in our construction is $O\left(\log G \right) = O(\log |T|)$. Moreover, we can use the fact that macro-tiles are all the same size in order to easily compute the indices $i_{-1} < i_0 < \cdots i_{k}$ in the definition of path simulate.

As shown in Figure~\ref{fig:zig-zag-TM-input-output}, each $t \in T$ has well-defined ``input'' and ``output'' sides (for convenience, and without loss of generality, we fix the seed row as growing from left to right), some of which may be the empty glue. The corresponding macro-tiles are depicted in Figure~\ref{fig:zig-zag-macro-tiles}. We encode each north/south glue in $T$ as a unique $(G + 1)$-bit binary string (we also encode the empty glue label, whence each glue is represented as a $(G + 1)$-bit binary string). We then encode each binary string (representing a glue) as a path of bumps and dents along the north and south side of the appropriate macro-tile(s). Into the bumps and dents, we carefully place the negative glue type to either represent a `0' or a `1' bit--similar to the construction for Theorem~\ref{square_theorem}. Note that we do not represent the ``east/west'' or strength-$2$ ``north/south'' glues of $T$ in this manner because in this case we encode these glue types in $T$ on the glues of the tiles which serve as the beginning and ends (inputs and outputs) of the paths forming the corresponding macro-tiles.

A macro-tile that represents a tile type $t \in T$ that binds via two ``input'' sides (e.g., south-east/south-west) self-assembles in two logical stages: reading the input glues and unpacking the output glues. In what follows, we will discuss the macro-tiles that represent tiles that have south/east input sides. The macro-tiles that emulate tiles with south/west input sides are constructed similarly.

\textbf{Reading the Input Glues.} In the first stage of the self-assembly of a ``south-east input'' macro-tile, an initial portion of its path crawls (either to the left or to the right) across the top of an existing macro tile. In doing so, the growing macro-tile path ``reads'' in, via a series of appropriate-placements of the negative glue in bumps and dents, a binary string, which represents a glue type in $T$. The method of ``reading'' a bit is depicted in Figure~\ref{fig:zig-zag-read-bit} and is similar to the technique used in Theorem~\ref{square_theorem} (see also Figure~\ref{fig:cooperation_gadget}).

\begin{figure}[htp]
    \begin{center}
    \includegraphics[width=5.0in]{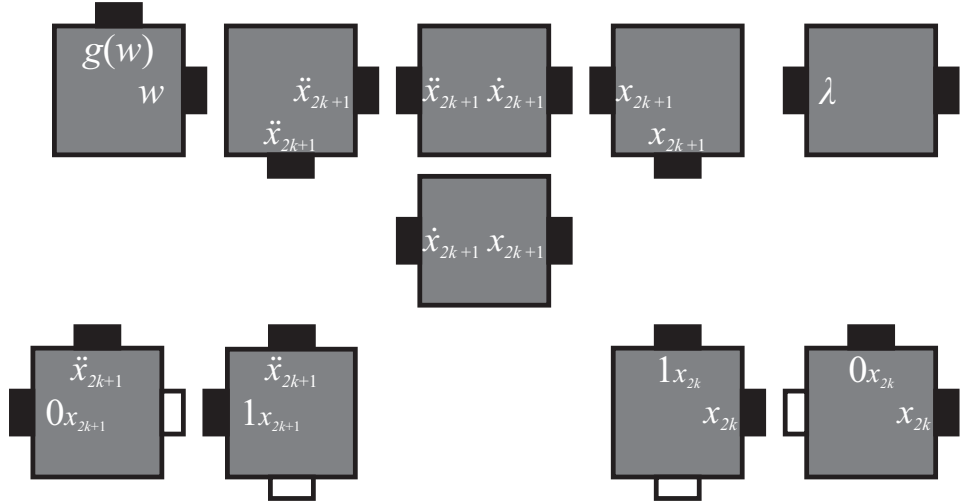} \caption{\label{fig:macro_tiles_read_binary} \small The tile types that read a binary string from the top of an existing macro-tile as they self-assemble from right to left. For each $i = 0, 1, \ldots, G + 1$, let $x_i \in \{0,1\}^i$. Note that the east input glue is implicitly encoded into the tile types shown here. The upper right tile initiates the process of reading the input binary string. The upper left tile completes the process by mapping a pair of south-east input glues to the corresponding north-west output glues. Tiles for south-west input macro tiles are designed similarly.}
    \end{center}
\end{figure}

For each $0 \leq i < G + 1$, we have a group of tile types that are responsible for collecting the $i^{\textmd{th}}$ bit of a binary string as they assemble a path to the left while ``remembering'' the previous $0 \leq j < i$ bits (see Figure~\ref{fig:zig-zag-reading-example}). Note that not every tile type in the group that reads the $i^{\textmd{th}}$ bit needs to remember \emph{all} $G + 1$ bits. In fact, it suffices for the group of tiles responsible for reading the $i^{\textmd{th}}$ bit to only remember $i$ bits. In order to do this, we use $O\left(2^0 + 2^1 + \cdots + 2^i\right) = O\left(2^{i + 1}\right)$ unique tile types, i.e., $O(1)$ tile types for each of the $2^j$ $j$-bit binary strings, whence we must have a total of $O\left(2^{\log (G + 1)}\right) = O(G) = O(|T|)$ unique tile types to read a glue from the top of an existing macro-tile (these tile types for a south-east input, north-west output macro-tile are shown in Figure~\ref{fig:macro_tiles_read_binary}). Once all $G + 1$ bits have been collected, we have a group of $O(G)$ unique tile types to convert the east input glue, along with the south input glue, into the appropriate output glue(s) for the macro-tile..

\begin{figure}[htp]
    \begin{center}
    \includegraphics[width=5.5in]{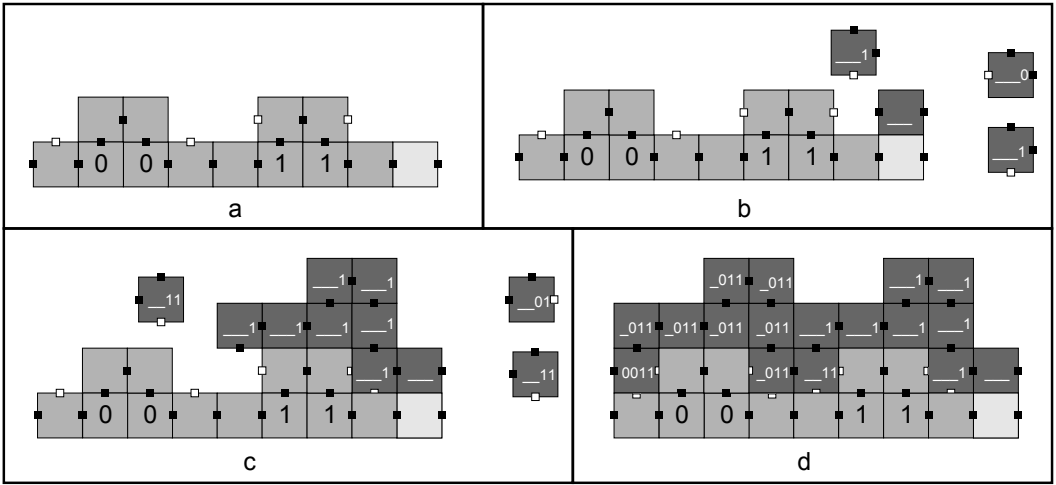} \caption{\label{fig:zig-zag-reading-example} \small An example depicting the north side of a macro-tile being ``read'' by the south side of another macro-tile. Here, the binary number being read is ``$0011$''.  The northern macro-tile grows from right to left.  Initially it has no information about the simulated glue to the south, and as it passes each position representing a bit, due to the configurations of the negative glue (pictured as white squares), it is able to place only one of two tiles, thus reading either a $0$ or $1$. See Figures~\ref{fig:macro_tiles_read_binary} and~\ref{fig:macro_tiles_write_binary} for more detail.}
    \end{center}
\end{figure}

\begin{figure}[htp]
    \begin{center}
    \includegraphics[width=4.5in]{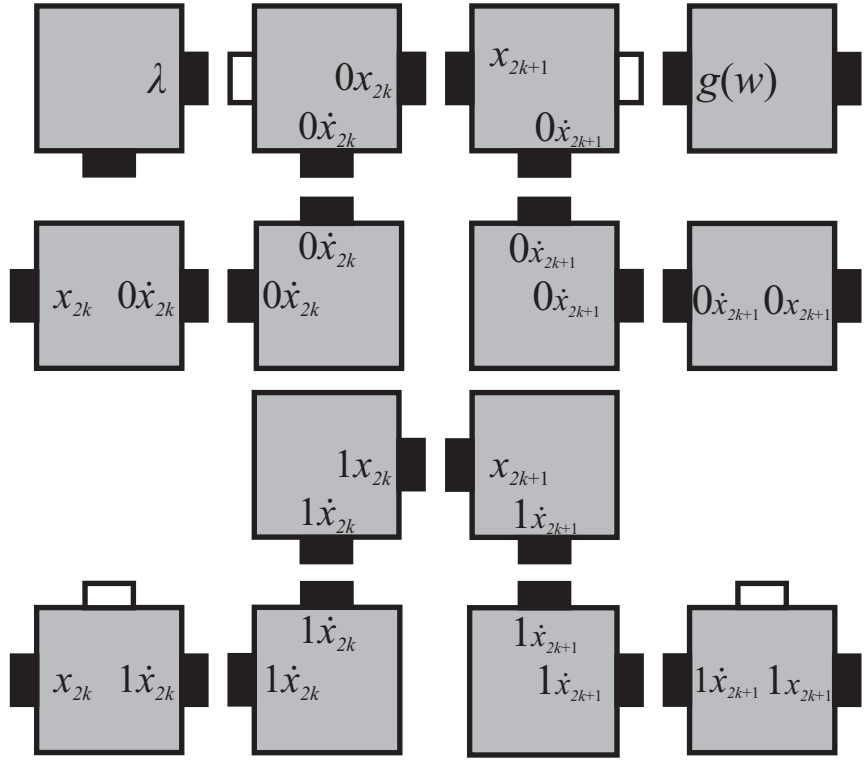} \caption{\label{fig:macro_tiles_write_binary} \small The tile types that unpack a glue type into binary string as they self-assemble from right to left. For each $i = 0, 1, \ldots, G + 1$, let $x_i \in \{0,1\}^i$. Note that the east input glue is implicitly encoded into the tile types shown here. The upper right tile initiates the process of unpacking the input binary string. The upper left tile marks the completion of this process. Tiles for south-west input macro tiles are designed similarly.}
    \end{center}
\end{figure}

\textbf{Unpacking the Output Glues.} After the output glue(s) of a macro-tile have been determined, the macro-tile path crawls back across itself and determines when to ``stop'' via the checkered tiles in Figure~\ref{fig:zig-zag-simulation}(b). Then the path crawls, once again, back across itself and ``unpacks'' the north output glue (it does not have to unpack the west output glue by the way we encode the east/west glue types in $T$ in macro-tiles). We accomplish this task in a manner that is similar to--but essentially the opposite of--reading in a $(G + 1)$-bit binary string. To do this, we use $O\left(G \right) = O(|T|)$ unique tile types (these tile types for a south-east input, north-west output macro tile are shown in Figure~\ref{fig:macro_tiles_write_binary}.

Finally, a macro-tile in $S$ that represents a tile type $t \in T$ that binds via a single, strength-$2$, input(output) side does not perform any input reading or output unpacking because we encode each strength-$2$ glue in $T$ as a unique strength-$1$ glue in $S$. Thus, when such a macro tile self-assembles, it does so in a single logical stage.



\begin{theorem}
\label{the_other_guys}
For every standard Turing machine $M$ and input $w$, the following hold.
\begin{enumerate}
    \item There exists an rgTAS that simulates $M$ on $w$
    \item There exists a prgTAS that simulates $M$ on $w$.
\end{enumerate}
\end{theorem}

\begin{proof}
For every Turing machine $M$ and $w \in \Sigma^*$, there exists a compact zig-zag TAS $\mathcal{T}_{M(w)}$ that simulates $M$ on input $w$ \cite{CookFuSch11}. The basic idea is to design $\mathcal{T}_{M(w)}$ so that self-assembly proceeds in a ``zig-zag'' growth pattern. This means that self-assembly proceeds according to a unique assembly sequence, which builds horizontal rows of tiles (configurations of $M$) one at a time, alternating growth from left-to-right and right-to-left. Figure~\ref{fig:zig-zag-TM} shows an example of a zig-zag Turing machine construction.

By Theorem~\ref{thm:zigzagSimulation1NegGluergTAS}, we can simulate $\mathcal{T}_{M(w)}$ with an rgTAS, and by Theorem~\ref{thm:zigzagSimulation1NegGlueprgTAS}, we can simulate $\mathcal{T}_{M(w)}$ with a prgTAS. 
\end{proof}

\begin{figure}[htp]
    \begin{center}
    \vspace{-5pt}
    \includegraphics[width=\textwidth]{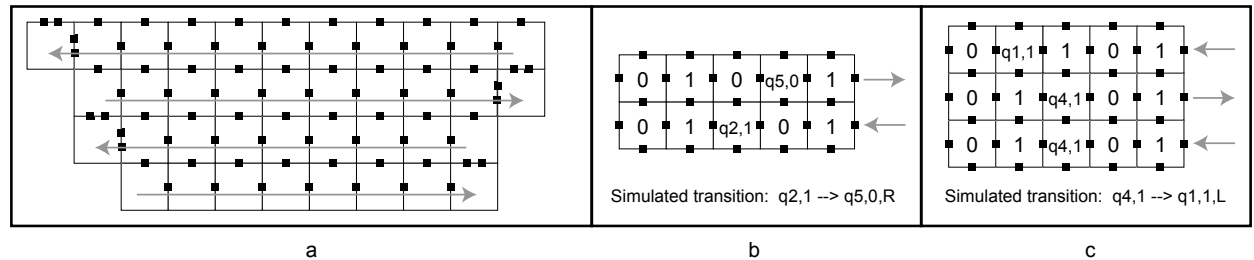} \caption{\label{fig:zig-zag-TM} \small Sketch of a zig-zag Turing machine. (a) Rows grow in alternating directions (grey arrows) and are extended in width by one tile per row.  Upward growth occurs only at the end of each row.  (b) Example of a TM transition which moves the head to the right occurring in a row growing left-to-right.  (c) Example of a TM transition which moves the head to the left.  Note that if this transition is encountered by a row which is growing left-to-right, the transition will be skipped in that row (an effective ``no op''), and instead simulated by the next row which grows right-to-left.}
    \end{center}
    \vspace{-30pt}
\end{figure} 

\subsubsection*{Acknowledgments}The authors would like to thank Dave Doty for insightful comments which inspired large improvements to the original results and led directly several of the current results.

\bibliographystyle{amsplain}
\bibliography{tam}

\providecommand{\bysame}{\leavevmode\hbox to3em{\hrulefill}\thinspace}
\providecommand{\MR}{\relax\ifhmode\unskip\space\fi MR }
\providecommand{\MRhref}[2]{%
  \href{http://www.ams.org/mathscinet-getitem?mr=#1}{#2}
}
\providecommand{\href}[2]{#2}
\begin{thebibliography}{10}

\bibitem{AdChGoHu01}
Leonard Adleman, Qi~Cheng, Ashish Goel, and Ming-Deh Huang, \emph{Running time
  and program size for self-assembled squares}, Proceedings of the thirty-third
  annual ACM Symposium on Theory of Computing (New York, NY, USA), ACM, 2001,
  pp.~740--748.

\bibitem{BarSchRotWin09}
Robert~D. Barish, Rebecca Schulman, Paul~W. Rothemund, and Erik Winfree,
  \emph{An information-bearing seed for nucleating algorithmic self-assembly},
  Proceedings of the National Academy of Sciences \textbf{106} (2009), no.~15,
  6054--6059.

\bibitem{CheSchGoeWin07}
Ho-Lin Chen, Rebecca Schulman, Ashish Goel, and Erik Winfree, \emph{Reducing
  facet nucleation during algorithmic self-assembly}, Nano Letters \textbf{7}
  (2007), no.~9, 2913--2919.

\bibitem{CookFuSch11}
Matthew Cook, Yunhui Fu, and Robert Schweller, \emph{Temperature 1
  self-assembly: Deterministic assembly in 3d and probabilistic assembly in
  2d}, Proceedings of the 22nd Annual ACM-SIAM Symposium on Discrete
  Algorithms, 2011.

\bibitem{Dot10}
David Doty, \emph{Randomized self-assembly for exact shapes}, {SIAM} Journal on
  Computing \textbf{39} (2010), no.~8, 3521--3552.

\bibitem{jDotKarMas10}
David Doty, Lila Kari, and Beno\^{i}t Masson, \emph{Negative interactions in
  irreversible self-assembly}, Algorithmica, to appear. Preliminary version
  appeared in DNA 2010.

\bibitem{USA}
David Doty, Jack~H. Lutz, Matthew~J. Patitz, Scott~M. Summers, and Damien
  Woods, \emph{Intrinsic universality in self-assembly}, Proceedings of the
  27th International Symposium on Theoretical Aspects of Computer Science,
  2009, pp.~275--286.

\bibitem{SFTSAFT}
David Doty, Matthew~J. Patitz, Dustin Reishus, Robert~T. Schweller, and
  Scott~M. Summers, \emph{Strong fault-tolerance for self-assembly with fuzzy
  temperature}, Proceedings of the 51st Annual IEEE Symposium on Foundations of
  Computer Science (FOCS 2010), 2010, pp.~417--426.

\bibitem{jLSAT1}
David Doty, Matthew~J. Patitz, and Scott~M. Summers, \emph{Limitations of
  self-assembly at temperature 1}, Theoretical Computer Science \textbf{412}
  (2011), 145--158.

\bibitem{KS07}
Ming-Yang Kao and Robert~T. Schweller, \emph{Reducing tile complexity for
  self-assembly through temperature programming}, Proceedings of the 17th
  Annual ACM-SIAM Symposium on Discrete Algorithms (SODA 2006), Miami, Florida,
  Jan. 2006, pp. 571-580, 2007.

\bibitem{KaoSchS08}
\bysame, \emph{Randomized self-assembly for approximate shapes}, International
  Colloqium on Automata, Languages, and Programming, Lecture Notes in Computer
  Science, vol. 5125, Springer, 2008, pp.~370--384.

\bibitem{DNAandMagnets}
Joseph~M. Kinsella and Albena Ivanisevic, \emph{Enzymatic clipping of dna wires
  coated with magnetic nanoparticles}, Journal of the American Chemical Society
  \textbf{127} (2005), no.~10, 3276--3277.

\bibitem{jSSADST}
James~I. Lathrop, Jack~H. Lutz, and Scott~M. Summers, \emph{Strict
  self-assembly of discrete {S}ierpinski triangles}, Theoretical Computer
  Science \textbf{410} (2009), 384--405.

\bibitem{MajumReif08}
Urmi Majumder and John Reif, \emph{A framework for designing novel magnetic
  tiles capable of complex self-assemblies}, Unconventional Computing (Cristian
  Calude, José Costa, Rudolf Freund, Marion Oswald, and Grzegorz Rozenberg,
  eds.), Lecture Notes in Computer Science, vol. 5204, Springer Berlin /
  Heidelberg, 2008, pp.~129--145.

\bibitem{MaoSunSee99}
Chengde Mao, Weiqiong Sun, and Nadrian~C. Seeman, \emph{Designed
  two-dimensional {D}{N}{A} holliday junction arrays visualized by atomic force
  microscopy.}, Journal of the American Chemical Society \textbf{121} (1999),
  no.~23, 5437--5443.

\bibitem{jReifSahuPeng06}
John~H. Reif, Sudheer Sahu, and Peng Yin, \emph{Complexity of graph
  self-assembly in accretive systems and self-destructible systems}, Theor.
  Comput. Sci. \textbf{412} (2011), no.~17, 1592--1605.

\bibitem{DNAandMagnets2}
David Rickwood and Vera Lund, \emph{Attachment of dna and oligonucleotides to
  magnetic particles: methods and applications}, Fresenius' Journal of
  Analytical Chemistry \textbf{330} (1988), 330--330, 10.1007/BF00469247.

\bibitem{Roth01}
Paul W.~K. Rothemund, \emph{Theory and experiments in algorithmic
  self-assembly}, Ph.D. thesis, University of Southern California, December
  2001.

\bibitem{RotWin00}
Paul W.~K. Rothemund and Erik Winfree, \emph{The program-size complexity of
  self-assembled squares (extended abstract)}, STOC '00: Proceedings of the
  thirty-second annual ACM Symposium on Theory of Computing (Portland, Oregon,
  United States), ACM, 2000, pp.~459--468.

\bibitem{RoPaWi04}
Paul~W.K. Rothemund, Nick Papadakis, and Erik Winfree, \emph{Algorithmic
  self-assembly of {DNA} {S}ierpinski triangles}, PLoS Biology \textbf{2}
  (2004), no.~12, 2041--2053.

\bibitem{SolWin07}
David Soloveichik and Erik Winfree, \emph{Complexity of self-assembled shapes},
  SIAM Journal on Computing \textbf{36} (2007), no.~6, 1544--1569.

\bibitem{Winf98}
Erik Winfree, \emph{Algorithmic self-assembly of {D}{N}{A}}, Ph.D. thesis,
  California Institute of Technology, June 1998.

\bibitem{WinLiuWenSee98}
Erik Winfree, Furong Liu, Lisa~A. Wenzler, and Nadrian~C. Seeman, \emph{Design
  and self-assembly of two-dimensional {D}{N}{A} crystals.}, Nature
  \textbf{394} (1998), no.~6693, 539--44.

\end{thebibliography}

\end{document}